\documentclass{amsart}
\usepackage{amssymb}
\usepackage{amsmath}
\usepackage{amsthm}
\usepackage{amscd}
\usepackage{amstext}
\usepackage{amsfonts}
\usepackage[all]{xy}
\usepackage{graphicx}
\usepackage{bm}

\theoremstyle{plain}
\newtheorem{theorem}{Theorem}[section]
\newtheorem{proposition}[theorem]{Proposition}
\newtheorem{lemma}[theorem]{Lemma}

\newtheorem{assumption}[theorem]{Assumption}
\theoremstyle{remark}
\newtheorem{definition}[theorem]{\rm Definition}
\newtheorem{example}[theorem]{\rm Example}
\newtheorem{remark}[theorem]{\rm Remark}

\newcommand{\Z}{\mathbb{Z}}
\newcommand{\R}{\mathbb{R}}
\newcommand{\C}{\mathbb{C}}
\newcommand{\T}{\mathbb{T}}
\newcommand{\U}{\mathbf{U}}

\newcommand{\I}{\mathcal{I}}
\newcommand{\Corner}{\mathrm{Corner}}
\newcommand{\BE}{\mathrm{Bulk-Edge}}
\newcommand{\HH}{\mathcal{H}}
\newcommand{\K}{\mathcal{K}}

\newcommand{\esssp}{\mathrm{ess\mathchar`-}\mathrm{sp}}
\newcommand{\AIII}{\mathrm{A\hspace{-.1em}I\hspace{-.1em}I\hspace{-.1em}I}}

\DeclareMathOperator{\End}{\mathrm{End}}
\DeclareMathOperator{\Ker}{\mathrm{Ker}}

\begin{document}
\title[Topological invariants and corner states]
{Topological invariants and corner states for Hamiltonians on a three-dimensional lattice}
\author[S. Hayashi]{Shin Hayashi}
\address{Mathematics for Advanced Materials -- Open Innovation Laboratory, AIST,  
c/o Advanced Institute for Materials Research, Tohoku University,  
2--1--1 Katahira, Aoba, Sendai 980--8577, Japan.}
\address{Graduate School of Mathematical Sciences, University of Tokyo, 3-8-1 Komaba, Tokyo, 153-8914, Japan.}
\email{{\tt shin-hayashi@aist.go.jp}}

\maketitle
\begin{abstract}
Periodic Hamiltonians on a three-dimensional (3-D) lattice with a spectral gap not only on the bulk but also on two edges at the common Fermi level are considered.
By using $K$-theory applied for the quarter-plane Toeplitz extension, two topological invariants are defined.
One is defined for the gapped bulk and edge Hamiltonians, and the non-triviality of the other means that the corner Hamiltonian is gapless.
A correspondence between these two invariants is proved.
Such gapped Hamiltonians can be constructed from Hamiltonians of 2-D type A and 1-D type AIII topological insulators, and its corner topological invariant is the product of topological invariants of these two phases.
\end{abstract}
\setcounter{tocdepth}{2}
\tableofcontents
\section{Introduction}
\label{sec:1}
In condensed matter physics, a correspondence between topological invariants defined for a gapped Hamiltonian of an infinite system without edge and that for a Hamiltonian of a system with edge is well known.
This correspondence is called the bulk-edge correspondence.
In the theoretical study of the quantum Hall effect, Thouless et al. \cite{TKNN82} introduced the integer valued topological invariant for a gapped Hamiltonian of an infinite system without edge. This invariant is called the TKNN number, and is equal to the first Chern number of the Bloch bundle \cite{Ko85}. 
Hatsugai considered such phenomena on a system with edge and defined a topological invariant as a winding number counted on a Riemann surface, in other words, the spectral flow of a family of self-adjoint Fredholm Toeplitz operators \cite{Hat93a}.
Hatsugai went on to prove the equality between these two topological invariants (bulk-edge correspondence) \cite{Hat93b}.

Non-commutative geometry was shown to be very useful in this context through the work of Bellissard \cite{Be86,BvES94}.
Kellendonk, Richter and Schulz-Baldes proved the bulk-edge correspondence by using the six-term exact sequence of $K$-Theory for $C^*$-algebras associated to the Toeplitz extension of the rotation algebra \cite{HKR00,HKR02}.

Analysis of Toeplitz operators have been further developed in additional studies (see \cite{Do98,BS06}, for example).
Here, we focus on the theory of quarter-plane Toeplitz algebras.
Such algebras were first studied by Douglas and Howe \cite{DH71,CDSS71,CDS72}.
Among the many results from their analyses, the following short exact sequence is most important in this paper.
\begin{equation}\label{seq1}
0 \to \K \to \mathcal{T}^{\alpha, \beta} \to \mathcal{S}^{\alpha, \beta} \to 0,
\end{equation}
where $\mathcal{T}^{\alpha, \beta}$ is the quarter-plane Toeplitz algebra (the symbols are defined below in Sect.~\ref{sec:2.3}).
This sequence was obtained first by Douglas and Howe for a special case, while Park proved the general case. He studied quarter-plane Toeplitz algebras by using $K$-theory for $C^*$-algebras \cite{Pa90,PS91}.

In this paper, we mainly consider a system with a codimension-two corner, which appears at the intersection of two codimension-one edges.
We consider a periodic Hamiltonian on the three-dimensional (3-D) lattice $\Z^3$, which has a spectral gap not only on the bulk (which means that our Hamiltonian is gapped) but also on two edges (which means that compressions of our Hamiltonian onto two subsemigroups on the lattice that correspond to two edges, are gapped) at the common Fermi level (Assumption~\ref{assumption}).
For such systems, we define two topological invariants. One is defined by two gapped edge Hamiltonians, which are the compressions of the common bulk (Definition~\ref{bulkedgeinv}), while the other is defined by corner Hamiltonians and is called the {\em gapless corner invariant} (Definition~\ref{cornerinv}). Both invariants are defined as elements of some $K$-groups. As a numerical invariant, the corner topological invariant is essentially defined by counting corner states.
We next show the correspondence of these two invariants (Theorem~\ref{main}), using the six-term exact sequence associated to (\ref{seq1}).
Note that for such bulk-edge gapped systems, some weak topological invariants are zero (Proposition~\ref{weak}).
In this sense, invariants considered in this paper can be seen as secondary invariants.
Although we mainly consider 3-D systems, our method can also be applied to systems of other dimensions (Remark~\ref{rem2}).

Such bulk-edge gapped Hamiltonians can be constructed from Hamiltonians of 2-D type A and 1-D type AIII topological insulators as follows.
Let $H_1$ be a Hamiltonian of a 2-D type A topological insulator and $H_2$ be that of a 1-D type AIII topological insulator whose chiral symmetry is given by $\Pi$.
Then, the bulk-edge gapped Hamiltonian is given by the following formula:
\begin{equation*}
	H = H_1 \bm{\otimes} \Pi + 1 \bm{\otimes} H_2.
\end{equation*}
Moreover, the spectral flow of the corner topological invariant is the product of the topological invariants of these two Hamiltonians (Theorem~\ref{product}).
Spectral gaps of two edge Hamiltonians constructed in this way remains open, unless those of the original two Hamiltonians closes. In this sense, our secondary topological phase of bulk-edge gapped Hamiltonians can be obtained as a product of 2-D type A and 1-D type AIII topological phases.
Using this construction, we obtain a non-trivial example (Example~\ref{example}).

Note that, except for the use of the sequence (\ref{seq1}), our invariants are defined and the correspondence is proved as in the case of the bulk-edge correspondence.
Although our study of such invariants is modeled on the Kellendonk--Richter--Schulz-Baldes proof of the bulk-edge correspondence \cite{HKR00,HKR02}, we note that there are many other proofs of bulk-edge correspondence from the perspective of $K$-theory and index theory \cite{EG02,GP13,ASV13,BCR1,Ku15,MT16,Hayashi}.

The bulk-edge correspondence for systems preserving some symmetries (odd time-reversal symmetry for quantum spin Hall systems, for example) is also well-known \cite{ASV13,GP13,Ku15,MT16,BKR16}.
In this sense, it is natural to expect such a ``bulk-edge and corner'' correspondence for systems preserving some symmetries.
Another direction will be higher codimensional cases, and the corresponding generalization of (\ref{seq1}) has already been discussed in \cite{DH71}.
These possible generalizations are not treated in this paper.

Studies of such bulk-edge gapped Hamiltonians and topologically protected corner states can also be found in the physical literature \cite{BBH17,HWK17}.
Benalcazar et al. \cite{BBH17} discussed a model that preserves some spatial symmetries and proposed experimental realizations.
Hashimoto et al. \cite{BBH17} considered different boundary conditions on two edges, and discussed such corner (edge-of-edge) states.
In this paper, we do not consider such spatial symmetry and our boundary condition is simply a Dirichlet boundary condition.

This paper is organized as follows.
In Sect.~\ref{sec:2}, we present some basic facts about $K$-theory for $C^*$-algebras, spectral flow and quarter-plane Toeplitz algebras.
Sect.~\ref{sec:3} fixes the conditions for the ``gapped'' Hamiltonians which are considered in this paper. Two topological invariants for such Hamiltonians are defined, and the relation between the two is proved.
In Sect.~\ref{sec:4}, we construct the gapped Hamiltonians from Hamiltonians of 2-D type A and 1-D type AIII topological insulators. The spectral flow of the corner topological invariant of this bulk-edge gapped Hamiltonian is shown to be equal to the product of topological invariants of these two Hamiltonians. A non-trivial example is given using this construction.

\section{Preliminaries}
\label{sec:2}
In this section, we present some necessary basic facts and calculations.
Throughout this paper, all algebras and Hilbert spaces are considered over the complex field $\C$, and all operators are complex linear.

\subsection{$K$-theory for $C^*$-algebras}
\label{sec:2.1}
We use $K$-theory for $C^*$-algebras in order to define some topological invariants and also to prove our main theorem.
This subsection gathers basic facts from $K$-theory for $C^*$-algebras without proof.
We refer the reader to \cite{Mu90,We93,Bl98,HR00,RLL00} for the details.

Let $A$ be a unital $C^*$-algebra, that is a Banach $*$-algebra with a multiplicative unit, which satisfies $||a^*a|| = ||a||^2$ for any $a$ in $A$.
An element $p$ in $A$ is called a projection if $p = p^* = p^2$, and an element $u$ in $A$ is said to be unitary if $u^*u = uu^* = 1$.
For a positive integer $n$, let $M_n(A)$ be the matrix algebra of all $n \times n$ matrices with entries in $A$.
As in the case of $M_n(\C)$, the matrix algebra $M_n(A)$ has a natural $*$-algebra structure. We know that $M_n(A)$ has a (unique) norm making it a $C^*$-algebra.
Let $P_n(A)$ be the set of all projections in $M_n(A)$, and let $P_\infty(A) := \bigsqcup_{n=1}^{\infty}P_n(A)$.
For $p \in P_n(A)$ and $q \in P_m(A)$, we write $p \sim_0 q$ if and only if there exists some $v \in M_{m,n}(A)$ such that $p = v^*v$ and $q = vv^*$.
The relation $\sim_0$ induces an equivalence relation on $P_\infty(A)$.
We define a binary operation $\oplus$ on $P_\infty(A)$ by $p \oplus q := \mathrm{diag}(p,q)$.
Then, $\oplus$ induces an addition $+$ on equivalence classes $D(A) := P_\infty(A)/\sim_0$, and $(D(A), +)$ is a commutative monoid.
The $K_0$-group $K_0(A)$ for the unital $C^*$-algebra $A$ is defined to be the group completion (Grothendieck group) of the commutative monoid $(D(A), +)$.
We denote the class of $p \in P_\infty(A)$ in $K_0(A)$ by $[p]_0$.
For a non-unital $C^*$-algebra $I$, we define its $K_0$-group $K_0(I)$ to be the kernel of the map $K_0(\tilde{I}) \to K_0(\C)$, where $\tilde{I} = I \oplus \C$ is the unitization of $I$, and the map is induced by the projection onto the second component.
Let $\mathcal{U}(A)$ be the group of unitary elements in $A$, and let $\mathcal{U}_n(A) := \mathcal{U}(M_n(A))$.
We consider the set $\mathcal{U}_\infty(A) := \bigsqcup_{n=1}^{\infty}\mathcal{U}_n(A)$ and a binary operation $\oplus$ on $\mathcal{U}_\infty(A)$, defined as above.
For $u \in \mathcal{U}_n(A)$ and $v \in \mathcal{U}_m(A)$, we write $u \sim_1 v$ if and only if there exists some $k \geq \max \{m,n \}$ such that $u \oplus 1_{k-n}$ and $v \oplus 1_{k-m}$ are homotopic in $\mathcal{U}_k(A)$.
The relation $\sim_1$ is an equivalence relation on $\mathcal{U}_\infty(A)$, and $\oplus$ induces an addition $+$ on equivalence classes $\mathcal{U}_\infty(A)/ \sim_1$.
Then, $(\mathcal{U}_\infty(A)/ \sim_1, +)$ is an abelian group. We denote this group by $K_1(A)$ and the class of $u \in \mathcal{U}_\infty(A)$ in $K_1(A)$ by $[u]_1$.
For a non-unital $C^*$-algebra $I$, its $K_1$-group is defined using its unitization, $K_1(I) := K_1(\tilde{I})$.
Note that by using the polar decomposition, an invertible element in $M_n(A)$ defines an element in $K_1(A)$.
$*$-homomorphisms $\varphi, \psi \colon A \to B$ between $C^*$-algebras are said to be homotopic if there exists a path of $*$-homomorphisms $\varphi_t \colon A \to B$ for $t \in [0,1]$ such that the map $[0,1] \to B$ defined by $t \mapsto \varphi_t(a)$ is continuous for each $a \in A$, $\varphi_0 = \varphi$ and $\varphi_1 = \psi$.
In this sense, $K_0$ and $K_1$ are the additive covariant homotopy functors from the category of $C^*$-algebras to the category of abelian groups.
Let $\mathcal{V}$ be a separable Hilbert space and $K(\mathcal{V})$ be the set of compact operators on $\mathcal{V}$.
For a $C^*$-algebra $A$, we have the stability property, that is, $K_i(K(\mathcal{V}) \bm{\otimes} A) \cong K_i(A)$, where $i=0,1$ and $\bm{\otimes}$ is the $C^*$-algebraic tensor product.

Let $A$ be a $C^*$-algebra. The suspension of $A$ is the $C^*$-algebra
$SA = \{ f \in C([0,1], A) \ | \ f(0) = f(1) = 0 \} \cong A \bm{\otimes} C_0((0,1))$, where $C([0,1], A)$ is the $C^*$-algebra of continuous functions from $[0,1]$ to $A$, and $C_0((0,1))$ is the $C^*$-algebra of complex valued continuous functions that vanish at infinity.
Then there exists an isomorphism $\theta_A \colon K_1(A) \to K_0(SA)$.
We also have a map $\beta_A \colon K_0(A) \to K_1(SA)$, called the Bott map.
If $A$ is a unital $C^*$-algebra, $\beta_A$ is given by $\beta_A[p]_0 = [f_p]_1$, where $f_p(t) = \exp(2\pi i t p)$.
By the Bott periodicity theorem, $\beta_A$ is an isomorphism.
For the short exact sequence of $C^*$-algebras, $0 \to I \to A \to B \to 0$, we can associate the following six-term exact sequence:
\[\xymatrix{
K_1(I) \ar[r]& K_1(A) \ar[r] & K_1(B) \ar[d]^{\delta_1}
\\
K_0(B) \ar[u]^{\delta_0} & K_0(A) \ar[l] & K_0(I). \ar[l]
}\]
The maps $\delta_0$ and $\delta_1$ are connecting homomorphisms for the sequence.
If $A$ and $B$ are unital $C^*$-algebras, and $I$ is a closed ideal in $A$, then $\delta_0$ is expressed in the following way.
For $p \in P_n(B)$, we can take its self-adjoint lift $\hat{p} \in M_n(A)$, and then we have $\delta_0[p]_0 = [\exp(-2\pi i \hat{p})]_1$.
The following diagram is commutative:
\begin{equation}\label{diag2}
\vcenter{
\xymatrix{
K_1(B) \ar[d]_{\delta_1} \ar[r]^{\theta_B} & K_0(SB) \ar[d]^{\delta_0} \\
K_0(I) \ar[r]^{\beta_I} & K_1(SI).}}
\end{equation}

\subsection{Spectral Flow and Winding Number}
\label{sec:2.2}
In this subsection, we discuss the relationship between the spectral flow and the winding number \cite{AS69,Ph96}.

Let $\T$ be a unit circle in the complex plane.
Let $\R_{>0}$ (resp. $\R_{<0}$) be the set of strictly positive (resp. negative) real numbers.
For a separable Hilbert space $\mathcal{V}$, we consider the space of bounded linear operators $B(\mathcal{V})$ on $\mathcal{V}$ equipped with the norm topology.
Let $\mathcal{F}^{s.a.}_\ast$ be the subspace of $B(\mathcal{V})$ consisting of all self-adjoint Fredholm operators whose essential spectra\footnote{In this paper, we denote the spectrum of an operator $T$ by $\mathrm{sp}(T)$, and the essential spectrum of $T$ by $\esssp(T)$.
For a bounded linear self-adjoint operator $T$, the essential spectrum of $T$ consists of accumulation points of $\mathrm{sp}(T)$ and isolated points of $\mathrm{sp}(T)$ with infinite multiplicity (see Proposition $2.2.2$ of \cite{HR00}, for example).} are not contained in either $\R_{>0}$ or $\R_{<0}$.
We consider the following subspace of $\mathcal{F}^{s.a.}_\ast$,
\begin{equation*}
	\hat{F}^\infty_\ast := \{ B \in \mathcal{F}^{s.a.}_\ast \ | \ ||B||=1, \ \mathrm{sp}(B) \ \text{is finite and} \ \esssp(B)= \{ \pm 1\} \}.
\end{equation*}
Let $i \colon \hat{F}^\infty_\ast \to \mathcal{F}^{s.a.}_\ast$ be the inclusion. Then, $i$ is a homotopy equivalence.
Let $\U(\infty)$ be the inductive limit of a sequence $\U(1) \to \U(2) \to \cdots$, where $\U(n)$ is the unitary group of degree $n$, and the map $\U(n) \to \U(n+1)$ is given by $\mathbf{A} \mapsto \mathrm{diag}(\mathbf{A}, \bm{1})$.
We have a map $j \colon \hat{F}^\infty_\ast \to \U(\infty)$ given by $j(B) = \exp(i \pi (B+1))$. The map $j$ is also a homotopy equivalence.\footnote{The subspace $\hat{F}^{\infty}_\ast$ and the commutativity of the diagram (\ref{sf}) are discussed explicitly in \cite{Ph96}, while the homotopy equivalence between $\hat{F}^{\infty}_\ast$ and $\U(\infty)$ was essentially proved in \cite{AS69}. Here, we follow the notations used in \cite{Ph96}.}
Thus, we have $K_1(C(\T)) = [\T, \U(\infty)] \cong [\T, \hat{F}^{\infty}_\ast] \cong [\T, \mathcal{F}^{s.a.}_\ast]$.
For a continuous loop in $\mathcal{F}^{s.a.}_\ast$, the spectral flow is defined. This is, roughly speaking, the net number of crossing points of eigenvalues with zero counted with multiplicity.
Then, the following diagram is commutative,
\begin{equation}\label{sf}
	\vcenter{
		\xymatrix{
[\T, \mathcal{F}^{s.a.}_\ast] \ar[r]^{\mathrm{sf}} & \Z
\\
[\T, \hat{F}^{\infty}_\ast] \ar[r]^{j_*} \ar[u]^{i_\ast} & [\T, \U(\infty)], \ar[u]}} \vspace{-1mm}
\end{equation}
where the map $\mathrm{sf} \colon [\T, \mathcal{F}^{s.a.}_\ast] \to \Z$ is given by the spectral flow, and the map $[\T, \U(\infty)] \to \Z$ is given by taking the winding number of the determinant (see Prop. $6$ of \cite{Ph96}).
All arrows indicate group isomorphisms.

\subsection{Quarter-Plane Toeplitz $C^*$-algebras}
\label{sec:2.3}
In this subsection, we present basic facts about quarter-plane Toeplitz algebras that will be used below \cite{DH71,Pa90}.

Let $\HH$ be the Hilbert space $l^2(\Z^2)$.
For a pair of integers $(m, n)$, let ${\bm e_{m,n}}$ be the element of $\HH$ that is $1$ at $(m,n)$ and $0$ elsewhere.
For such $(m, n)$, let $M_{m,n} \colon \HH \to \HH$ be the translation operator defined by $(M_{m,n}\varphi)(k,l) = \varphi(k-m, l-n)$.
We choose real numbers $\alpha < \beta$, and let $\HH^{\alpha}$ and  $\HH^{\beta}$ be closed subspaces of $\HH$ spanned by $\{ {\bm e_{m,n}} \ | \ -\alpha m + n \geq 0 \}$ and $\{ {\bm e_{m,n}} \ | \ -\beta m + n \leq 0 \}$, respectively.
Let $\HH^{\alpha, \beta}$ be their intersection $\HH^{\alpha} \cap \HH^{\beta}$(see Figure $1$.).
\begin{figure}
\centering
    \includegraphics[width=84mm,clip]{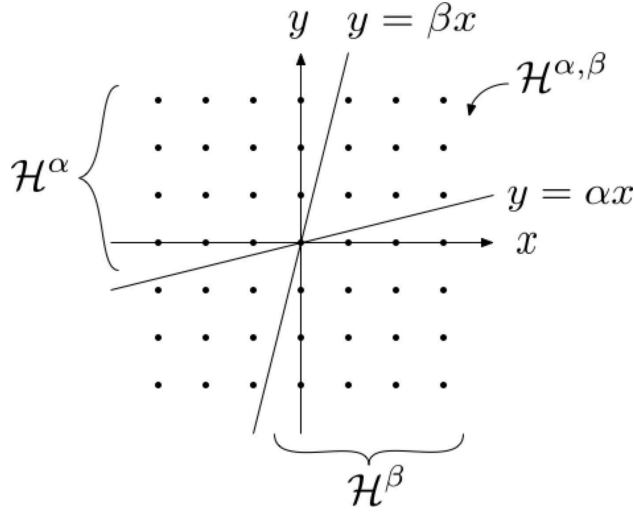}
\caption{Half planes $\HH^{\alpha}$ and $\HH^{\beta}$, and the quarter-plane $\HH^{\alpha, \beta}$}
\end{figure}
We can take $\alpha = - \infty$ or $\beta = + \infty$, but not both.
Let $P^\alpha$ and $P^\beta$ be orthogonal projections of $\HH$ onto $\HH^\alpha$ and $\HH^\beta$, respectively.
Then, $P^\alpha P^\beta$ is the orthogonal projection of $\HH$ onto $\HH^{\alpha, \beta}$.
Note that $P^\alpha$ and $P^\beta$ commute.
We define the {\em quarter-plane Toeplitz $C^*$-algebra} to be the $C^*$-subalgebra $\mathcal{T}^{\alpha, \beta}$ of $B(\HH^{\alpha, \beta})$ generated by $\{ P^\alpha P^\beta M_{m,n}P^\alpha P^\beta \ | \ (m,n) \in \Z^2 \}$.
We also define the {\em half-plane Toeplitz $C^*$-algebras} $\mathcal{T}^\alpha$ and $\mathcal{T}^\beta$ to be $C^*$-subalgebras of $B(\HH^{\alpha})$ and $B(\HH^{\beta})$ generated by $\{ P^\alpha M_{m,n} P^\alpha \ | \ (m,n) \in \Z^2 \}$ and $\{ P^\beta M_{m,n} P^\beta \ | \ (m,n) \in \Z^2 \}$, respectively\footnote{Let $\mathcal{V}$ be a Hilbert space, $\mathcal{W} \subset \mathcal{V}$ be its closed subspace, $P_{\mathcal{W}}$ be an orthogonal projection of $\mathcal{V}$ onto $\mathcal{W}$ and $T \colon \mathcal{V} \rightarrow \mathcal{V}$ be a bounded linear operator. We write $P_{\mathcal{W}} T P_{\mathcal{W}}$ for the compression of $T$ onto $\mathcal{W}$, that is, the operator on $\mathcal{W}$ of the form $\varphi \mapsto P_{\mathcal{W}} T \varphi$ ($\varphi \in \mathcal{W}$).}.
We have surjective $*$-homomorphisms $\gamma^\alpha \colon \mathcal{T}^{\alpha, \beta} \to \mathcal{T}^\alpha$ and $\gamma^\beta \colon \mathcal{T}^{\alpha, \beta} \to \mathcal{T}^\beta$, which map $P^\alpha P^\beta M_{m,n}P^\alpha P^\beta$ to $P^\alpha M_{m,n} P^\alpha$ and $P^\beta M_{m,n}P^{\beta}$, respectively.
We also have surjective $*$-homomorphisms $\sigma^\alpha \colon \mathcal{T}^\alpha \to C(\T^2)$ and $\sigma^\beta \colon \mathcal{T}^\beta \to C(\T^2)$ which map $P^\alpha M_{m,n} P^\alpha$ to $\chi_{m,n}$ and $P^\beta M_{m,n} P^\beta$ to $\chi_{m,n}$, respectively, where $\chi_{m,n}(z_1, z_2) = z_1^m z_2^n$.
The well-definedness of $\gamma^\alpha$ and $\gamma^\beta$ is proved in \cite{Pa90}, and that of $\sigma^\alpha$ and $\sigma^\beta$ is proved in \cite{CD71}.
Let $\rho^\alpha \colon \mathcal{T}^\alpha \to \mathcal{T}^{\alpha,\beta}$, $\rho^\beta \colon \mathcal{T}^\beta \to \mathcal{T}^{\alpha,\beta}$ and $\xi^{\alpha, \beta} \colon C(\T^2) \to \mathcal{T}^{\alpha,\beta}$ be bounded linear maps given by compression; that is, $\rho^\alpha(X) = P^\beta X P^\beta$, $\rho^\beta(Y) = P^\alpha Y P^\alpha$ and $\xi^{\alpha,\beta}(Z) = P^\alpha P^\beta Z P^\alpha P^\beta$, respectively.
We define a $C^*$-algebra $\mathcal{S}^{\alpha, \beta}$ to be the pullback of $\mathcal{T}^\alpha$ and $\mathcal{T}^\beta$ along $C(\T^2)$; that is,
$\mathcal{S}^{\alpha, \beta} := \{ (T^\alpha, T^\beta) \in \mathcal{T}^\alpha \oplus \mathcal{T}^\beta \ | \ \sigma^\alpha(T^\alpha) = \sigma^\beta(T^\beta) \}$.
We write $p^\alpha \colon \mathcal{S}^{\alpha,\beta} \rightarrow \mathcal{T}^\alpha$ and $p^\beta \colon \mathcal{S}^{\alpha,\beta} \rightarrow \mathcal{T}^\beta$ for $*$-homomorphisms given by projections onto each component.
There is a $*$-homomorphism $\gamma \colon \mathcal{T}^{\alpha, \beta} \to \mathcal{S}^{\alpha, \beta}$ given by $\gamma(T) = (\gamma^\alpha(T), \gamma^\beta(T))$.
Let $\K$ be the $C^*$-algebra of compact operators on $\mathcal{H}^{\alpha,\beta}$.
\begin{theorem}[Douglas--Howe\cite{DH71}, Park\cite{Pa90}]\label{exact}
There is the following short exact sequence,
\begin{equation*}
0 \to \K \to \mathcal{T}^{\alpha, \beta} \overset{\gamma}{\to} \mathcal{S}^{\alpha, \beta} \to 0,
\end{equation*}
which has a linear splitting\footnote{Note that the map $\rho$ is just a linear map and not a $*$-homomorphism. This sequence splits as a short exact sequence of linear spaces, but does not split as that of $C^*$-algebras.}
 $\rho \colon \mathcal{S}^{\alpha, \beta} \to \mathcal{T}^{\alpha, \beta}$ given by
$\rho(T^\alpha, T^\beta) = \rho^\alpha(T^\alpha) +\rho^\beta(T^\beta) - \xi^{\alpha,\beta}\sigma^\beta(T^\beta)$.
\end{theorem}
The following theorem follows immediately by Atkinson's theorem.\footnote{Such results were first obtained by Douglas and Howe \cite{DH71} for the special case $\mathcal{T}^{0,\infty}$ by expressing the algebra $\mathcal{T}^{0,\infty}$ as a tensor product of two Toeplitz algebras. Park obtained the Fredholm condition for the general $\mathcal{T}^{\alpha,\beta}$ in a different way \cite{Pa90}.}
\begin{theorem}[Douglas-Howe\cite{DH71}, Park\cite{Pa90}]\label{Fredholm}
An operator $T$ in $\mathcal{T}^{\alpha, \beta}$ is a Fredholm operator if and only if $\gamma^\alpha(T)$ and $\gamma^\beta(T)$ are both invertible elements in $\mathcal{T}^\alpha$ and $\mathcal{T}^\beta$, respectively.
\end{theorem}

By taking a tensor product of the above sequence and $C(\T)$, we have the short exact sequence,\footnote{Since $C(\T)$ is an abelian $C^*$-algebra, $C(\T)$ is a nuclear $C^*$-algebra by Takesaki's theorem. Since $C(\T)$ is nuclear, this sequence is exact (see \cite{Mu90}, for example).}
$0 \to \K \bm{\otimes} C(\T) \to \mathcal{T}^{\alpha, \beta} \bm{\otimes} C(\T) \to \mathcal{S}^{\alpha, \beta} \bm{\otimes}C(\T) \to 0$.
Associated to this sequence, we have the following six-term exact sequence:
\[\xymatrix{
K_1(\mathcal{K} \bm{\otimes} C(\T)) \ar[r]& K_1(\mathcal{T}^{\alpha,\beta} \bm{\otimes} C(\T)) \ar[r] & K_1(\mathcal{S}^{\alpha,\beta} \bm{\otimes} C(\T)) \ar[d]^{\delta_1}
\\
K_0(\mathcal{S}^{\alpha,\beta} \bm{\otimes} C(\T)) \ar[u]^{\delta_0} & K_0(\mathcal{T}^{\alpha,\beta} \bm{\otimes} C(\T)) \ar[l] & K_0(\mathcal{K}\bm{\otimes} C(\T)), \ar[l]
}\]
where $\delta_0$ and $\delta_1$ are connecting homomorphisms associated to this short exact sequence.
By the stability property, $K_1(\mathcal{K} \bm{\otimes} C(\T))$ is isomorphic to $K_1(C(\T))$.
\begin{lemma}\label{surj}
	The map $\delta_0 \colon K_0(\mathcal{S}^{\alpha,\beta} \bm{\otimes} C(\T)) \to K_1(C(\T))$ is surjective.
\end{lemma}
\begin{proof}
Let us take a base point of $\T$. Then we have the isomorphisms
$K_0(\mathcal{S}^{\alpha,\beta} \bm{\otimes} C(\T))
	\cong
		K_0(\mathcal{S}^{\alpha,\beta}) \oplus K_1(\mathcal{S}^{\alpha,\beta})$
and
$K_1(\mathcal{K} \bm{\otimes} C(\T))
	\cong
		K_1(\mathcal{K}) \oplus K_0(\mathcal{K})
			\cong
				0 \oplus \Z$.
Consider the following commutative diagram,\footnote{Horizontal arrows are the composite of (inverse) maps induced by homomorphisms of $C^*$-algebras and suspension isomorphisms (see Example $7.5.1$ of \cite{Mu90}), and the commutativity of this diagram follows from the naturality of the connecting homomorphism $\delta_0$ and diagram (\ref{diag2}).}
\[\xymatrix{
K_0(\mathcal{S}^{\alpha,\beta} \bm{\otimes} C(\T)) \ar[r]^{\cong \hspace{3mm}} \ar[d]^{\delta_0} &
	 	 K_0(\mathcal{S}^{\alpha,\beta}) \oplus K_1(\mathcal{S}^{\alpha,\beta}) \ar[d]^{\delta_0 \oplus \delta_1}
\\
K_1(\mathcal{K} \bm{\otimes} C(\T)) \ar[r]^{\cong} & K_1(\mathcal{K}) \oplus K_0(\mathcal{K}).
}\]
The group $K_1(\mathcal{S}^{\alpha,\beta})$ is isomorphic to $\Z$ \cite{Pa90},
and the map $\delta_1 \colon K_1(\mathcal{S}^{\alpha,\beta}) \to K_0(\mathcal{K})$ is an isomorphism \cite{Pa90,Ji95}.
Thus the left map is surjective.
\end{proof}
\begin{remark}\label{rem1}
Note that the group $K_0(\mathcal{S}^{\alpha, \beta})$ is calculated as follows \cite{Pa90}.
$$
K_0(\mathcal{S}^{\alpha, \beta}) \cong
\left\{
\begin{aligned}
\Z^{ \ } & \hspace{3mm} \text{if $\alpha$ and $\beta$ are both rational,}\\
\Z^2 & \hspace{3mm} \text{if either $\alpha$ or $\beta$ is rational and the other is irrational,}\\
\Z^3 & \hspace{3mm} \text{if $\alpha$ and $\beta$ are both irrational.}
\end{aligned}
\right.
$$
The case of $\alpha = - \infty$ (or $\beta = \infty$) is the same as that of rational $\alpha$ (or rational $\beta$).
From this calculation, we can see that $K_0(\mathcal{S}^{\alpha, \beta})$ depends delicately on $\alpha$ and $\beta$, which correspond to the angles of edges.
By the proof of Lemma \ref{surj}, this component maps to zero by $\delta_0$.

Generators of the group $K_0(\mathcal{S}^{\alpha, \beta})$ are given as follows.
In any cases, one direct summand $\Z$ is generated by the identity.
For the second case, let us consider the case of rational $\alpha = p/q$ where $p$, $q$ are relatively prime integers, and $\beta$ is irrational. We take relatively prime integers $r$ and $s$ such that $p/q \neq r/s$. Let $\hat{\chi} = P^\beta \chi_{r,s} P^\beta$, then, the other $\Z$ direct summand is generated by the class $[(1 - \hat{\chi}^*\hat{\chi}, 0)]_0 - [(1-\hat{\chi}\hat{\chi}^*,0)]_0$.
The third case ($\alpha$ and $\beta$ are both irrational) has two $\Z$ direct summands generated by such elements.
\end{remark}

\section{Bulk-Edge and Corner Correspondence}
\label{sec:3}
In this section, we consider some ``gapped'' Hamiltonians and define two topological invariants for them.
The correspondence between them is also proved.

\subsection{Topological invariants for bulk-edge gapped Hamiltonians}
\label{sec:3.1}
Let $V$ be a finite dimensional Hermitian vector space, and denote the complex dimension of $V$ by $N$.
Let $\HH_V := \HH \bm{\otimes} V$, $\HH_V^\alpha := \HH^\alpha \bm{\otimes} V$, $\HH_V^\beta := \HH^\beta \bm{\otimes} V$, $\HH_V^{\alpha,\beta} := \HH^{\alpha,\beta}\bm{\otimes} V$, $P^\alpha_V := P^\alpha \bm{\otimes} 1$ and $P^\beta_V := P^\beta \bm{\otimes} 1$.
We consider a continuous map $\T^3 \rightarrow \End_\C(V)$, $(\xi, \eta, t) \mapsto H(\xi, \eta, t)$, where, for each $(\xi,\eta,t) \in \T^3$, $H(\xi,\eta,t)$ is Hermitian.
The multiplication operator\footnote{Let $f \colon \T^3 \to \End_\C(V)$ be a continuous map. Then the operator on $L^2(\T^3; V)$ defined by $g \mapsto fg$ is called the multiplication operator generated by $f$.} generated by $H(\xi, \eta, t)$ defines an operator on $L^2(\T^3; V)$.
Through the Fourier transform, we obtain a bounded linear self-adjoint operator $H$ on $l^2(\Z^3; V)$.
We call $H$ the {\em bulk Hamiltonian}.
By the Fourier transform in the last $\Z$ component, we obtain a continuous family of bounded linear self-adjoint operators $\{ H(t) \colon \HH_V \rightarrow \HH_V \}_{t \in \T}$.
Using this, we consider the following half-plane Toeplitz operators:
\begin{equation*}
	H^\alpha(t) := P^\alpha_V H(t) P^\alpha_V \colon \HH_V^\alpha \to \HH_V^\alpha,
\hspace{2mm}
	H^\beta(t) := P^\beta_V H(t) P^\beta_V \colon \HH_V^\beta \to \HH_V^\beta.
\end{equation*}
They are bounded self-adjoint operators.
We call $H^\alpha(t)$ and $H^\beta(t)$ {\em edge Hamiltonians}.
Let us take an orthonormal frame of $V$.
Then, since $H^\alpha(t)$ and $H^\beta(t)$ are compressions of the same operator $H(t)$, the pair $(H^\alpha(t), H^\beta(t))$ defines a self-adjoint element of the $C^*$-algebra $M_N(\mathcal{S}^{\alpha, \beta} \bm{\otimes} C(\T))$.
Let $\mu$ be a real number.
The following is our assumption.
\begin{assumption}[Spectral Gap Condition]\label{assumption}
We assume that our edge Hamiltonians have a common spectral gap at the Fermi level $\mu$ for any $t$ in $\T$, i.e.,
$\mu$ is not contained in either $\mathrm{sp}(H^\alpha(t))$ or $\mathrm{sp}(H^\beta(t))$.
We call this condition as a {\em spectral gap condition}.
\end{assumption}
\begin{remark}\label{sgc}
Such $\mu$ does exist. Actually, we can make $\mu$ sufficiently large or small.
However, if we choose such $\mu$, our topological invariants will be zero (see also Remark~\ref{general}).
Non-trivial invariants appear, at least if $\mu$ lies in a common spectral gap of $H^\alpha(t)$ and $H^\beta(t)$.
Note that in this case, our bulk Hamiltonian $H(t)$ also has a spectral gap at $\mu$ since $\mathrm{sp}(H(t))$ is contained in $\mathrm{sp}(H^\alpha(t), H^\beta(t))$.\footnote{There is a $*$-homomorphism $M_N(\mathcal{S}^{\alpha,\beta}) \rightarrow M_N(C(\T^2))$ which maps $(H^\alpha(t), H^\beta(t))$ to $H(t)$.}
\end{remark}

In what follows, we assume $\mu = 0$ without loss of generality.
Following Remark \ref{sgc}, we further assume that $\mathrm{sp}(H^\alpha(t), H^\beta(t))$ is not contained in either $\R_{>0}$ or $\R_{<0}$.\footnote{This is a technical assumption, and turns out to be not very important (see Remark \ref{general}).}
Note that under our spectral gap condition, the element $(H^\alpha(t), H^\beta(t))$ is invertible in $M_N(\mathcal{S}^{\alpha, \beta} \bm{\otimes} C(\T))$.

We next define topological invariants for our bulk-edge gapped Hamiltonians.
Let $h \colon \R \setminus \{0\} \rightarrow \R$ be a continuous function which is $1$ on $\R_{< 0}$ and $0$ on $\R_{> 0}$.
\begin{definition}\label{bulkedgeinv}
By the continuous functional calculus (see \cite{Mu90}, for example), we have a projection $p := h(H^\alpha(t), H^\beta(t)) = (h(H^\alpha(t)), h(H^\beta(t)))$ in $M_N(\mathcal{S}^{\alpha,\beta} \bm{\otimes} C(\T))$.
We denote the element $[p]_0$ in the $K$-group $K_0(\mathcal{S}^{\alpha,\beta} \bm{\otimes} C(\T))$ by $\I_\BE(H)$.\footnote{In order to define the element $[p]_0$ of the $K$-group $K_0(\mathcal{S}^{\alpha,\beta} \bm{\otimes} C(\T))$, we took the orthonormal frame of $V$. However, this element does not depend on the choice.}
\end{definition}

We here compare $\I_\BE(H)$ and well-known bulk invariants.
Our bulk Hamiltonian $H$ defines an element $[h(H(\xi,\eta,t))]_0$ of the $K$-group $K_0(C(\T^3)) \cong \Z \oplus \Z \oplus \Z \oplus \Z$.\footnote{From the perspective of topological $K$-theory, the group $K_0(C(\T^3))$ is isomorphic to the topological $K$-cohomology group $K^0(\T^3)$, and the element $[h(H(\xi,\eta,t))]_0$ corresponds to the class of the Bloch bundle.}
One direct summand $\Z$ is generated by the identity, and topological invariants for the bulk Hamiltonian corresponding to the other three components are called {\em weak invariants}.
\begin{proposition}\label{weak}
For our gapped Hamiltonians, these weak invariants are zero.
\end{proposition}
\begin{proof}
The algebra $\mathcal{S}^{\alpha, \beta}$ is the pullback of $\mathcal{T}^\alpha$ and $\mathcal{T}^\beta$ along $C(\T^2)$.
Applying the Mayer-Vietoris sequence in $K$-theory as in \cite{Pa90}, we can calculate the map $(\sigma^\alpha \circ p^\alpha)_* \colon K_i(\mathcal{S}^{\alpha,\beta}) \to K_i(C(\T^2))$ for $i = 0,1$.
$K$-groups $K_0(\mathcal{S}^{\alpha,\beta})$ and $K_0(C(\T^2))$ both have a $\Z$ direct summand generated by the identity.
Since $\sigma^\alpha \circ p^\alpha$ is the unital $*$-homomorphism, the induced homomorphism $(\sigma^\alpha \circ p^\alpha)_*$ maps this component isomorphically.
By calculating the Mayer-Vietorius sequence, we see that $(\sigma^\alpha \circ p^\alpha)_* \colon K_0(\mathcal{S}^{\alpha,\beta}) \to K_0(C(\T^2))$ maps the other components to zero (generators of $K_0(\mathcal{S}^{\alpha,\beta})$ are given in Remark \ref{rem1}). We also see that $(\sigma^\alpha \circ p^\alpha)_* \colon K_1(\mathcal{S}^{\alpha,\beta}) \to K_1(C(\T^2))$ is the zero map.

Using the isomorphism $K_0(\mathcal{S}^{\alpha,\beta} \bm{\otimes} C(\T)) \cong K_0(\mathcal{S}^{\alpha,\beta}) \oplus K_1(\mathcal{S}^{\alpha,\beta})$, we see that the map
$(\sigma^\alpha \circ p^\alpha \bm{\otimes} 1)_* \colon K_0(\mathcal{S}^{\alpha,\beta} \bm{\otimes} C(\T)) \to K_0(C(\T^3))$
maps isomorphically the $\Z$ direct summand generated by the identity, and that the other components map to zero.
The result then follows, since
$(\sigma^\alpha \circ p^\alpha \bm{\otimes} 1)_* ([p]_0) = [h(H(\xi,\eta,t))]_0$.
\end{proof}

\subsection{Gapless corner invariant}
\label{sec:3.2}
We next consider the following quarter-plane Toeplitz operators,
\begin{equation*}
	H^{\alpha,\beta}(t) := P^\alpha_V P^\beta_V H(t) P^\alpha_V P^\beta_V \colon \HH_V^{\alpha,\beta} \to \HH_V^{\alpha,\beta}.
\end{equation*}
We call $H^{\alpha,\beta}(t)$ {\em corner Hamiltonians}.

\begin{definition}\label{cornerinv}
By the spectral gap condition and Theorem~\ref{Fredholm}, we have a continuous family $\{ H^{\alpha,\beta}(t) \}_{t \in \T}$ of bounded linear self-adjoint Fredholm operators.
This family defines an element $\I_\Corner(H)$ of the $K$-group $K_1(C(\T)) \cong [\T, \mathcal{F}^{s.a.}_\ast]$.
We call $\I_\Corner(H)$ the {\em gapless corner invariant}.
\end{definition}

\begin{remark}\label{numcorner}
By diagram (\ref{sf}), the map $\mathrm{sf} \colon K_1(C(\T)) \to \Z$ maps the gapless corner invariant $\I_\Corner(H)$ to the spectral flow of the family $\{ H^{\alpha,\beta}(t) \}_{t \in \T}$.
Thus, if the gapless corner invariant is non-trivial, our corner Hamiltonian is gapless, and there exist topologically protected {\em corner states}.
Note that the invariant $\I_\BE(H)$ does not change as long as the spectral gaps of two edge Hamiltonians remain open.
This stability also holds for the gapless corner invariant.
\end{remark}

\subsection{Correspondence}
\label{sec:3.3}
The following is the main theorem of this paper.
\begin{theorem}\label{main}
The map $\delta_0 \colon K_0(\mathcal{S}^{\alpha,\beta} \bm{\otimes} C(\T)) \to K_1(C(\T))$ maps $\I_\BE(H)$ to the gapless corner invariant.
That is,
$\delta_0(\I_\BE(H)) = \I_\Corner(H)$.
\end{theorem}
\begin{proof}
Let $q := (1-h)(H^\alpha(t), H^\beta(t))$.
Since $\delta_0([p]_0 + [q]_0) = \delta_0[1_N]_0 = 0$, we have $\delta_0(\I_\BE(H)) = - \delta_0[q]_0 = [\exp(2\pi i \hat{q})]_1$, where $\hat{q} := \rho(q)$ is a self-adjoint lift of $q$.
By the spectral gap condition and Theorem~\ref{exact}, we have $\esssp(\hat{q}(t)) = \mathrm{sp}(q(t)) = \{ 0, 1 \}$.
By considering a spectral deformation that collapses eigenvalues in some small neighborhoods of $0$ and $1$ to points $0$ and $1$, respectively, we can deform $\hat{q}(t)$ into an element $\tilde{q}(t)$ of $\hat{F}^\infty_\ast$.
Thus, we obtain $[\exp(2\pi i \hat{q})]_1 = [\exp(2\pi i \tilde{q})]_1$.
Let us consider the isomorphisms $K_1(C(\T)) = [\T, \U(\infty)] \cong [\T, \hat{F}^{\infty}_\ast] \cong [\T, \mathcal{F}^{s.a.}_\ast]$.
We then obtain the following relation.
\begin{equation*}
 \delta_0(\I_\BE(H)) = [\exp(2\pi i \hat{q})]_1 = [2\tilde{q} - 1] = [2\hat{q} - 1].
\end{equation*}
Since two loops $\{ 2\tilde{q}(t) - 1 \}_{t \in \T}$ and $\{ H^{\alpha,\beta}(t) \}_{t \in \T}$ are homotopic in $\mathcal{F}^{s.a.}_\ast$, we have $[2\tilde{q} - 1] = [\{ H^{\alpha,\beta}(t) \}_{t \in \T}] = \mathcal{I_\Corner}(H)$.
\end{proof}
\begin{remark}\label{general}
If the spectrum of $(H^\alpha(t), H^\beta(t))$ is contained in $\R_{>0}$ or $\R_{<0}$, then we can define a topological invariant $\I_\BE(H)$ in the same way.
The spectral flow of the family $\{ H^{\alpha,\beta}(t) \}_{t \in \T}$ is also defined.\footnote{In this case, the family $\{ H^{\alpha,\beta}(t)\}_{t \in \T}$ is not contained in $\mathcal{F}^{s.a.}_\ast$ and does not define an element of the $K$-group $K_1(C(\T))$.}
In this case, $\delta_0(\I_\BE(H)) = 0$ and $\mathrm{sf}(\{ H^{\alpha,\beta}(t)\}_{t \in \T}) = 0$ holds.
Thus, we still have ``bulk-edge and corner'' correspondence.
\end{remark}
\begin{remark}\label{rem2}
Since $\T$ is just a parameter space in our formulation, we can generalize the parameter space $\T$ to other spaces.
Let $X$ be a compact Hausdorff space, and consider a continuous family $\{ H(x) \colon \HH_V \rightarrow \HH_V \}_{x \in X}$ of bounded self-adjoint multiplication operators generated by continuous maps.
We then define two edge Hamiltonians and the corner Hamiltonian in the same way.
If we assume our spectral gap condition, we can define the topological invariant $\I_\BE(H)$ and the gapless corner invariant $\mathcal{I_\Corner}(H)$ in the same way as elements of $K_0(\mathcal{S}^{\alpha,\beta} \bm{\otimes} C(X))$ and $K_1(C(X))$, respectively.
The exponential map $\delta_0$ maps $\I_\BE(H)$ to the gapless corner invariant $\mathcal{I_\Corner}(H)$ as in Theorem \ref{main}.
In this case, we lack the understanding of the gapless corner invariant as a spectral flow, but we still have a relation between our invariants and corner states; that is, if the gapless corner invariant is non-trivial, then there exist topologically protected corner states.
In particular, if we take $X = \T^2$, then this argument gives a ``bulk-edge and corner'' correspondence for such $4$-D systems.
\end{remark}

\section{Product Formula}
\label{sec:4}
In this section, we present one method for constructing a Hamiltonian with a spectral gap on the bulk and two edges using periodic Hamiltonians of a $2$-D type A topological phase and a $1$-D type AIII topological phase.
We see that such construction leads to some kind of ``product formula'' of these two topological phases.
By using this construction, we give one explicit non-trivial example of bulk-edge gapped Hamiltonians of non-trivial topological invariants.
In this section, we consider the case of $\alpha = 0$, $\beta = \infty$ and $\mu = 0$.

Let $V_1$ and $V_2$ be finite rank Hermitian vector spaces.
We assume that $V_2$ has a $\Z_2$-grading given by a $\C$-linear map $\Pi \colon V_2 \to V_2$ with $\Pi^2 = 1$.
We extend $\Pi$ to a $\Z_2$-grading on the Hilbert space $l^2(\Z; V_2)$ by the point-wise operation. We also write $\Pi$ for this $\Z_2$-grading.
Let $H_1$ be a multiplication operator on $l^2(\Z^2; V_1)$ generated by a continuous map $\T^2 \to \End_\C(V_1)$.
We assume that $H_1$ is self-adjoint and invertible.
In the classification table \cite{SRFL08,Kit09}, $H_1$ is in $2$-D type A.
Its partial Fourier transform gives a continuous family of bounded linear self-adjoint operators $\{H_1(t) \}_{t \in \T}$ on the Hilbert space $l^2(\Z; V_1)$.
We write $\I^{2d, \mathrm{A}}(H_1) \in \Z$ for the topological invariant of this system\footnote{$\I^{2d, \mathrm{A}}(H_1)$ is the minus of the first Chern number of the Bloch bundle, called the TKNN number \cite{TKNN82,PS16}. The sign convention used here is described in great detail in \cite{Hayashi}.}, 
Let $H_2$ be a multiplication operator on $l^2(\Z; V_2)$ generated by a continuous map $\T \to \End_\C(V_2)$.
We assume that $H_2$ is self-adjoint and invertible, and that $H_2$ satisfies $\Pi H_2 \Pi^* = - H_2$ (chiral symmetry).
In the classification table \cite{SRFL08,Kit09}, $H_2$ is in $1$-D type AIII.
We write $\I^{1d, \AIII}(H_2) \in \Z$ for the topological invariant of this system \cite{PS16}.
By the bulk-edge correspondence \cite{PS16}, these two topological invariants are the same as the one for systems with edge, and we here take this edge picture.
Let $P_i$ be the orthogonal projection of $l^2(\Z, V_i)$ onto $l^2(\Z_{\geq 0}, V_i)$ $(i=1,2)$.
Then $\I^{2d, \mathrm{A}}(H_1)$ is given by the minus of the spectral flow of the family of the self-adjoint Fredholm operators $\{ P_1H_1(t)P_1 \}_{t \in \T}$,
and $\I^{1d, \AIII}(H_2)$ is given by the minus of the signature of $\Pi$ restricted onto $\Ker P_2H_2P_2$.

Using these operators, let us consider the following bounded linear self-adjoint operator $H$ on the Hilbert space $l^2(\Z^3; V_1 \bm{\otimes} V_2)$,
\begin{equation}\label{prodHam}
	H = H_1 \bm{\otimes} \Pi + 1 \bm{\otimes} H_2.
\vspace{-1mm}
\end{equation}
Its partial Fourier transform gives a family of bounded linear self-adjoint operators $\{ H(t) = H_1(t) \bm{\otimes} \Pi + 1 \bm{\otimes} H_2 \}_{t \in \T}$ on the Hilbert space $l^2(\Z^2; V_1 \bm{\otimes} V_2)$.
\begin{theorem}\label{product}
For such a Hamiltonian $H$, the following holds.
\begin{itemize}
	\item[(1)] For any $t \in \T$, the operators $H(t)$, $H^0(t)$ and $H^\infty(t)$ are invertible.
	\item[(2)] We have $\mathrm{sf}(\I_\Corner(H)) = \I^{2d, \mathrm{A}}(H_1) \cdot \I^{1d, \AIII}(H_2)$,
where the right hand side is the product of two integers.
\end{itemize}
\end{theorem}

\noindent
{\it Proof of Theorem~\ref{product} (1)}. \
Here, we show only the invertibility of $H^0(t)$. The others can be proved in the same way.
Note that,
\begin{equation*}
	H^0(t) = P_1H_1(t) P_1 \bm{\otimes} \Pi + P_1 \bm{\otimes} H_2,
\end{equation*}
and its square is given as follows.
\begin{equation*}
	(H^0(t))^2 = (P_1H_1(t) P_1)^2 \bm{\otimes} 1 + P_1 \bm{\otimes} (H_2)^2.
\end{equation*}
By our assumption, $(H_2)^2$ is invertible.
Since positive operators form a cone, $(H^0(t))^2$ is positive and invertible.
Thus $H^0(t)$ is invertible.

\

By (1) of Theorem \ref{product}, the gapless corner invariant $\I_\Corner(H)$ is defined.
Note that the spectral gap of bulk and two edge Hamiltonians does not close if the spectral gaps of two bulk Hamiltonians $H_1(t)$ and $H_2$ remain open.

\

\noindent
{\it Proof of Theorem \ref{product} (2)}. \ 
It is enough to show the following equality,
\vspace{-1mm}
\begin{equation*}
	\mathrm{sf}(\{ P_1H_1(t)P_1 \}_{t \in \T}) \cdot \mathrm{sign}(\Pi|_{\Ker P_2H_2P_2}) = \mathrm{sf}(\{ H^{0,\infty}(t) \}_{t \in \T}),
\end{equation*}
where $\mathrm{sign}(\Pi|_{\Ker P_2H_2P_2})$ is the signature of the operator $\Pi|_{\Ker P_2H_2P_2}$.
Note that
$H^{0,\infty}(t) = P_1H_1(t)P_1 \bm{\otimes} P_2 \Pi P_2 + P_1 \bm{\otimes} P_2H_2 P_2$.
If $P_2H_2P_2$ is invertible, we have $\I^{1d, \AIII}(H_2) = 0$.
Moreover, as in the proof of Theorem \ref{product} (1), it follows that $H^{0,\infty}(t)$ is invertible for any $t \in \T$.
Thus, we have $\I_\Corner(H) = 0$.

We next consider the case when the operator $P_2H_2P_2$ has a non-trivial kernel.
Let $k$ and $l$ be the number of $+1$ and $-1$ eigenvalues of the operator $\Pi|_{\Ker P_2H_2P_2}$, respectively.
Then $\mathrm{sign}(\Pi|_{\Ker P_2H_2P_2}) = k-l$.
As in the proof of Theorem \ref{product} (1), we have
\begin{equation}\label{eq}
	\{ t \in \T \ | \ 0 \in \mathrm{sp}(P_1H_1(t)P_1) \}
		=
	\{ t \in \T \ | \ 0 \in \mathrm{sp}(H^{0,\infty}(t)) \},
\end{equation}
and
\begin{equation}\label{eq2}
	\Ker H^{0,\infty}(t) = \Ker P_1H_1(t)P_1 \bm{\otimes} \Ker P_2H_2P_2.
\end{equation}
Since $\I^{2d, \mathrm{A}}(H_1)$ depends only on the homotopy class of $\{H_1(t)\}_{t \in \T}$ in $[\T,\mathcal{F}^{s.a.}_\ast]$, we can assume, if necessary, that set (\ref{eq}) is a finite set.
In the following, we assume this condition, and write $\{ t_1, \cdots, t_m \}$ for this set.
For each $i \in \{ 1, \cdots, m \}$, there is a crossing point of $\mathrm{sp}(P_1H_1(t)P_1)$ with $0 \in \R$ at $t_i$.
We write $n_i$ for the rank of $\Ker P_{\geq 0}H_1(t_i)P_{\geq 0}$.
If this is a positive crossing, then the contribution of this crossing point to the spectral flow $\mathrm{sf}(\{ P_1H_1(t)P_1 \}_{t \in \T})$ is $+n_i$.
In this case, by (\ref{eq}), there also is a crossing point of $\mathrm{sp}(H^{0,\infty}(t))$ with $0 \in \R$ at $t_i$.
Note that for $\varphi_1 \in \Ker(H_1(t) - z)$ and $\varphi_2 \in \Ker H_2$, we have
$H(t)\varphi_1 \bm{\otimes} \varphi_2 = z \varphi_1 \bm{\otimes} \Pi \varphi_2$.
Thus, by equation (\ref{eq2}), the contribution of this crossing point to the spectral flow $\mathrm{sf}(\{ H^{0,\infty}(t) \}_{t \in \T})$ is $n_i k - n_i l = n_i (k-l)$.
In the case of negative crossing, its contribution to $\mathrm{sf}(\{ P_1H_1(t)P_1 \}_{t \in \T})$ and $\mathrm{sf}(\{ H^{0,\infty}(t) \}_{t \in \T})$ is $-n_i$ and $-n_i (k-l)$, respectively.
By taking their sum with respect to $i$, the desired equality follows.

\

We here give an explicit example of a bulk-edge gapped Hamiltonian and calculate its gapless corner invariant.
\begin{example}\label{example}
Let $\sigma_1$, $\sigma_2$ and $\sigma_3$ be the following Pauli matrices:
\begin{equation*}
\sigma_1 =
\begin{pmatrix}
0 & 1\\
1 & 0
\end{pmatrix}, \
\sigma_2 =
\begin{pmatrix}
0 & -i\\
i & 0
\end{pmatrix}, \
\sigma_3 =
\begin{pmatrix}
1 & 0\\
0 & -1
\end{pmatrix}.
\end{equation*}
Let $H_1$ be the following bounded linear self-adjoint operator on $l^2(\Z^2) \bm{\otimes} \C^2$.
\begin{equation*}
	H_1 = \frac{1}{2i} \sum_{j=1}^2 (S_j - S_j^*) \bm{\otimes} \sigma_j + (-1 + \frac{1}{2} \sum_{j=1}^2(S_j + S_j^*)) \bm{\otimes} \sigma_3,
\end{equation*}
where $S_1 = M_{1,0}$ and $S_2 = M_{0,1}$ are translation operators.
This Hamiltonian is taken from Sect. 2.2.4 of \cite{PS16}.
$H_1$ is an example of a $2$-D type A topological insulator.
Its topological invariant is calculated in \cite{PS16} and is $\I^{2d, \mathrm{A}}(H_1) = -1$.
Let $H_2$ and $\Pi$ be the following bounded linear self-adjoint operators on the Hilbert space $l^2(\Z) \bm{\otimes} \C^2$:
\begin{equation*}
	H_2 = \frac{1}{2} S \bm{\otimes} (\sigma_1 + i \sigma_2) + \frac{1}{2} S^* \bm{\otimes} (\sigma_1 - i \sigma_2), \
	\Pi = 1 \bm{\otimes} \sigma_3,
\end{equation*}
where $S$ is the translation operator given by $(S \varphi)(n) = \varphi(n-1)$.
Note that the operator $H_2$ satisfies the chiral symmetry; that is, $\Pi H_2 \Pi^* = - H_2$.
$H_2$ is an example of a $1$-D type AIII topological insulator.
Its topological invariant is calculated in \cite{PS16} and is $\I^{1d, \AIII}(H_2) = -1$.
We consider a bounded linear self-adjoint operator $H = H_1 \bm{\otimes} \Pi + 1 \bm{\otimes} H_2$ on $l^2(\Z^3, \C^4)$.
Note that $V_1 = V_2 = \C^2$ in this case.
By Theorem \ref{product} (1), $H$ satisfies our spectral gap condition. 
By Theorem \ref{product} (2), we have $\mathrm{sf}(\I_\Corner(H)) = \I^{2d, \mathrm{A}}(H_1) \cdot \I^{1d, \AIII}(H_2) = (-1) \cdot (-1) =1$.
Note that, by the proof of Theorem \ref{product} (2), the corner states of $H^{0,\infty}(t)$ are given by the tensor product of two edge states of $P_1H_1(t)P_1$ and $P_2H_2P_2$, which are calculated in \cite{PS16}.
\end{example}

\begin{remark}
The topological invariant $\I_\BE(H)$ for the Hamiltonian $H$ of (\ref{prodHam}) is calculated as follows.
In this case $\alpha = 0$ and $\beta = \infty$, so that the $K$-group where this invariant lives is calculated as $K_0(\mathcal{S}^{0,\infty} \bm{\otimes} C(\T)) \cong \Z \oplus \Z$.
The map to the first $\Z$ direct summand is given by taking the rank of the Bloch bundle over the three-dimensional torus.
The map to the second $\Z$ direct summand is given by taking the spectral flow of the gapless corner invariant.

Let $k_1$ be the rank of the Bloch bundle associated to $H_1$, and let $M_2$ be the rank of complex vector space $V_2$, which is necessarily an even integer by the chiral symmetry.
For our Hamiltonian $H$, the integer corresponding to the first $\Z$ direct summand is given by the rank of the Bloch bundle associated to $H$, which is calculated by $k_1M_2$.
Thus, the topological invariant $\I_\BE(H)$ for the Hamiltonian of the form (\ref{prodHam}) is calculated as $\mathcal{I}_\BE(H) = (k_1M_2, \I^{2d, \mathrm{A}}(H_1) \cdot \I^{1d, \AIII}(H_2))$.
In the case of Example \ref{example}, we have $\mathcal{I}_\BE(H) = (2,1)$.
\end{remark}

\subsubsection*{Acknowledgements}
This work is part of a Ph.D. thesis, defended at the University of Tokyo in 2017.
The author would like to expresses his gratitude for the support and encouragement of his supervisor, Mikio Furuta.
This work was inspired by the author's collaborative research with Mikio Furuta, Motoko Kotani, Yosuke Kubota, Shinichiroh Matsuo and Koji Sato.
He would like to thank them for many stimulating conversations and much encouragements.
He also would like to thank Christopher Bourne, Ken-Ichiro Imura, Takeshi Nakanishi and Yukinori Yoshimura for their many discussions, and thank Emil Prodan for sharing the information regarding \cite{BBH17}. 
This work was supported by JSPS KAKENHI Grant Number JP17H06461.

\end{document}